\documentclass[12pt,reqno]{amsart}
\usepackage[left=1.25in,right=1.25in,top=1.25in,bottom=1.25in]{geometry}
\usepackage{color}
\usepackage[dvipsnames]{xcolor}
\usepackage{amsmath, amsthm, amssymb, amsfonts}
\usepackage{natbib}
\usepackage{har2nat}
\usepackage{tabularx}
\usepackage{parskip}
\usepackage[british]{babel}
\usepackage{mathtools}
\usepackage{xfrac}
\usepackage[colorlinks=true,allcolors=Blue,linkcolor=Blue]{hyperref}
\usepackage{dsfont}
\usepackage[onehalfspacing]{setspace}

\usepackage{tikz}
\usepackage{caption}
\usepackage{pgfplots}
\usepackage{float}
\usepackage{verbatim}
\restylefloat{table}
\usepgfplotslibrary{fillbetween}
\usetikzlibrary{patterns}
\pgfplotsset{compat=newest}
\usetikzlibrary{decorations.pathreplacing,angles,quotes}
\usepackage{graphicx}
\usepackage{subcaption}
\usepackage{threeparttable}
\allowdisplaybreaks
\newtheorem{theorem}{Theorem}
\newtheorem{corollary}{Corollary}
\theoremstyle{definition}

\newcommand{\ul}{\underline}
\newcommand{\ol}{\overline}
\newcommand{\df}{\mathrm{d}}

\newcommand{\bdis}{\begin{displaymath}}
\newcommand{\edis}{\end{displaymath}}
\newcommand{\beq}{\begin{equation}}
\newcommand{\eeq}{\end{equation}}
\newcommand{\bea}{\begin{eqnarray*}}
\newcommand{\eea}{\end{eqnarray*}}
\newcommand{\bean}{\begin{eqnarray}}
\newcommand{\eean}{\end{eqnarray}}

\newcommand{\R}{\mathbb{R}}

\DeclareMathOperator*{\argmax}{arg\,max}

\newcommand{\1}{\mathbf{1}}

\begin{document}

\let\MakeUppercase\relax 

\title{Distributions of Posterior Quantiles via Matching}
\author[\uppercase{Kolotilin and Wolitzky}]{\larger \textsc{Anton Kolotilin and Alexander Wolitzky}}
\date{\today}
\thanks{\\
\textit{Kolotilin}: School of Economics, UNSW Business School. \\
\textit{Wolitzky}: Department of Economics, MIT}

\begin{abstract} 
We offer a simple analysis of the problem of choosing a statistical experiment to optimize the induced distribution of posterior medians, or more generally $q$-quantiles for any $q \in (0,1)$. We show that all implementable distributions of the posterior $q$-quantile are implemented by a single experiment, the \emph{$q$-quantile matching experiment}, which pools pairs of states across the $q$-quantile of the prior in a positively assortative manner, with weight $q$ on the lower state in each pair. A dense subset of implementable distributions of posterior $q$-quantiles can be uniquely implemented by perturbing the $q$-quantile matching experiment. A linear functional is optimized over distributions of posterior $q$-quantiles by taking the optimal selection from each set of $q$-quantiles induced by the $q$-quantile matching experiment. The $q$-quantile matching experiment is the only experiment that simultaneously implements all implementable distributions of the posterior $q$-quantile.\\

\bigskip

\noindent\emph{JEL\ Classification:}\ C61, D72, D82\newline


\end{abstract}

\maketitle
\thispagestyle{empty}

\let\MakeUppercase\relax %

\newpage
\clearpage
\pagenumbering{arabic} 

\section{Introduction}
\label{s:intro}

Several problems of recent economic interest amount to characterizing the set of distributions of posterior quantiles that can be induced by some statistical experiment, or to finding a distribution in this set that maximizes some objective. These problems include \emph{apparent overconfidence} \citep{BD}---e.g., what distributions of medians of individuals' beliefs about their own abilities are consistent with Bayesian updating?; \emph{partisan gerrymandering} (\citealt{FH}; \citealt{KW})---e.g., what is the highest distribution of district median voters attained by any districting plan?; and \emph{quantile persuasion} \citep{KW20}---e.g., what experiment maximizes the expected action of a receiver who minimizes the expected absolute deviation of her action from the unknown state of the world?\footnote{\citet{YangZentefis} explore these and other applications. \citet{KW} consider a more general gerrymandering model, which reduces to optimizing the distribution of posterior quantiles in a special case. \citet{KW20} introduce quantile persuasion as a special case of a more general persuasion model, which is further developed in \citet{KCW23}.
}

We provide a simple solution for this problem.  For any $q \in (0,1)$, there is a single experiment---the \emph{$q$-quantile matching experiment}---that simultaneously implements all implementable distributions of the posterior $q$-quantile.  For example, if the state is uniformly distributed on $[0,1]$ and the relevant quantile is the median, the $q$-quantile matching experiment is the \emph{median-matching experiment} that, whenever the true state is $\theta \in [0,1/2]$, reveals only that the state is either $\theta$ or $1/2+\theta$.\footnote{To our knowledge, the median-matching experiment first appears in \citet{KW20}. It is closely related to the median one-to-one matching introduced by \citet{KM} and further studied by \citet{LN}. }  In general, the $q$-quantile matching experiment pools pairs of states across the $q$-quantile of the prior in a positively assortative manner, with weight $q$ on the lower state in each pair.

\begin{figure}
\begin{tikzpicture}
	\begin{axis}		
		[axis x line = middle,
		axis y line = middle,
		xmin = 0, xmax = 9/8,
		ymin = 0, ymax = 9/8,
		xlabel=$x$,
		ylabel=$H$,
		xtick={1/128,1/4,1/2,3/4,1},
		xticklabels={$0$,$\frac{1}{4}$,$\frac{1}{2}$,$\frac{3}{4}$,$1$},
		ytick={1/128, 1/4,1/2, 3/4,1},
		yticklabels={$0$,$\frac{1}{4}$,$\frac{1}{2}$, $\frac{3}{4}$,$1$},
		clip=false]
		
		\draw [very thick, solid, red]		(0,0) -- (1/2, 1);
		\draw [very thick, solid, red]		(1/2,1) -- (1, 1);
		\draw [very thick, solid, blue]		(1/2,0) -- (1, 1);
		\draw [very thick, solid, blue]		(0,0) -- (1/2, 0);
		\draw [thick, solid, black]			(0, 0) -- (1, 1);
		
		
		\draw [thin, dotted, black]			(0,0) -- (1/2,0);
		\draw [thin, dotted, black]			(1/8,1/4) -- (5/8,1/4);
    	\draw [thin, dotted, black]			(1/4,1/2) -- (3/4,1/2);
    	\draw [thin, dotted, black]			(3/8,3/4) -- (7/8,3/4);
    	\draw [thin, dotted, black]			(1/2,1) -- (1,1);
    	\draw [very thin, loosely dashed, black]			(1/2,0) -- (1/2,1);
    	
    	\node [circle, fill = red, scale=0.3]		at (0, 0)	{};
    	\node [circle, fill = red, scale=0.3]		at (1/8, 1/4)	{};	
    	\node [circle, fill = red, scale=0.3]		at (1/4, 1/2)	{};	
    	\node [circle, fill = red, scale=0.3]		at (3/8, 3/4)	{};	
    	\node [circle, fill = red, scale=0.3]		at (1/2, 1)	{};	
    	
    	\node [circle, fill = blue, scale=0.3]		at (1/2, 0)	{};
    	\node [circle, fill = blue, scale=0.3]		at (5/8, 1/4)	{};	
    	\node [circle, fill = blue, scale=0.3]		at (3/4, 1/2)	{};	
    	\node [circle, fill = blue, scale=0.3]		at (7/8, 3/4)	{};	
    	\node [circle, fill = blue, scale=0.3]		at (1, 1)	{};	
    	
    	\node [black] at (6/16,7/16) {$F$};
    	\node [red] at (3/16,1/2) {$\ol H$};
    	\node [blue] at (13/16,1/2) {$\ul H$};
    	   	
%

\end{axis}
\end{tikzpicture}
\caption{Implementable Distributions of Posterior Medians}
\caption*{\emph{Notes:} When the prior $F$ is uniform on $[0,1]$, $\ol H$ and $\ul H$ are the lowest and highest implementable distributions of posterior medians. A distribution $H$ is implementable iff $\ul H \leq H \leq \ol H$. 
Optimizing a linear functional over distributions of posterior medians requires taking the optimal selection from each horizontal dotted line. For example, the blue (red) dots are the optimal selections for an increasing (decreasing) objective function.}
\label{f:DP}
\end{figure}

To see why the $q$-quantile matching experiment implements all implementable distributions of the posterior $q$-quantile, consider again the median-matching experiment with a uniform state.  When the experiment reveals that the state is $\theta$ or $1/2+\theta$ with equal probability, every value $x\in [\theta,1/2+\theta]$ is a posterior median.  In particular, the median-matching experiment simultaneously implements (i) the distribution $\ul H(x)=\max \{0,2x-1\}$, (ii) the distribution $\ol H(x)=\min \{2x, 1\}$, and (iii) every distribution $H$ satisfying $\ul H \leq H \leq \ol H$.  Conversely, simple Markov-type inequalities imply that every implementable distribution is bounded by $\ul H$ and $\ol H$.  Moreover, the set of implementable distributions of unique posterior quantiles is essentially the same: any desired selection from each set of $q$-quantiles induced by the $q$-quantile matching experiment can be uniquely selected by assigning it an extra $\varepsilon$ probability.  Finally, optimizing a linear functional over distributions of posterior quantiles simply requires taking the optimal selection from each set of $q$-quantiles induced by the $q$-quantile matching experiment.  See Figure 1 for an illustration of our results.\footnote{Similar figures in the literature include Figure 1 of \citet{OG}, Figure 2 of \citet{KG}, and Figure 3 of \citet{YangZentefis}.}

We also show that the $q$-quantile matching experiment is the \emph{unique} experiment that implements all implementable distributions of the posterior $q$-quantile. To see why, consider again a uniform state, and compare the median-matching experiment with the \emph{negative assortative matching} experiment that, whenever the true state is $\theta \in [0,1/2]$, reveals only that the state is either $\theta$ or $1-\theta$. The negative assortative matching experiment simultaneously implements the lowest and highest distributions of the posterior median, $\ol H$ and $\ul H$, but it does not implement all intermediate distributions, such as the distribution $H_{1/2}$ given by $H_{1/2}(x)=\ol H(x)$ for $x<1/4$, $H_{1/2}(x)=1/2$ for $x \in [1/4,3/4)$, and $H_{1/2}(x)=\ul H(x)$ for $x\geq 3/4$. Indeed, the negative assortative matching experiment induces posteriors with medians between $1/4$ and $3/4$ when the true state lies between $1/4$ and $3/4$, while $H_{1/2}$ assigns probability $0$ to these medians.

The current paper is closely related to \citet{BD} and \citet{YangZentefis}. Both of these papers establish results that are very similar to our Theorem \ref{t:implementable} (albeit \citeauthor{BD} do so for discrete experiments with finitely many induced posteriors). Our main contribution is to introduce the $q$-quantile matching experiment, which yields a much simpler proof of Theorem \ref{t:implementable}, as well as new results (Theorems \ref{t:opt} and \ref{t:uniqueness}).

\section{Implementable Distributions of Posterior Quantiles}
\label{s:implementable}

This section shows that the $q$-quantile matching experiment implements all implementable distributions of the posterior $q$-quantile.

Let $\Theta=[\ul \theta,\ol \theta]\subset \R $, with $\ul \theta<\ol \theta$, be a compact state space; let $C(\Theta)$ be the set of continuous functions on $\Theta$; let $\Delta(\Theta)$ be the set of cumulative distribution functions on $\Theta$, endowed with the weak$^\star$ topology, and let $\Delta(\Delta(\Theta))$ be the set of probability measures on $\Delta(\Theta)$. Recall that $G\in \Delta(\Theta)$ is a non-decreasing, right-continuous function satisfying $G(\ul \theta)\geq 0$ and $G(\ol \theta)=1$. Let $\delta_x$, with $x\in \Theta$, denote the degenerate distribution at $x$, so that $\delta_x(\theta)=\1\{\theta\geq x\}$.

Fix a prior distribution $F\in \Delta(\Theta)$ and a quantile of interest $q\in (0,1)$. Following \citet{KG}, define an \emph{experiment} as a distribution $\tau\in \Delta(\Delta (\Theta))$ of posterior distributions $G \in \Delta(\Theta)$ such that $\int G \df \tau (G)=F$. For each posterior $G$, define the set of \emph{$q$-quantiles} of $G$ as 
\[
X(G)=\{x\in \Theta :G(x^-)\leq q \leq G(x)\},
\]
where $G(x^-)$ denotes the left limit $\lim_{\theta \uparrow x} G(\theta)$, with the convention $G(\ul \theta^-)=0$.
In addition, for each $G$, define its generalized inverse $G^{-1}$ as
\[
G^{-1}(p)=\inf \{\theta \in \Theta: G(\theta)\geq p\}, \quad \text{for all $p\in [0,1]$}.
\]
That is, $G^{-1}(p)$ is the smallest $p$-quantile of $G$. 

To define the $q$-quantile matching experiment, let $\omega$ be uniformly distributed on $[0,1]$, and for each $\omega \in [0,q]$, let $G=G_\omega$ be the distribution that assigns probability $q$ to $F^{-1}(\omega)$ and assigns probability $1-q$ to $F^{-1}(q +(1-q)\omega/q)$.  
The \emph{$q$-quantile matching experiment} is defined as an experiment $\tau^{\star}$ such that for $\tau^{\star}$-almost all $G$, there exists $\omega \in [0,q]$ such that $G=G_\omega$.\footnote{For example, when $F$ is atomless, we can let $\omega = F(\theta)$, so that the $q$-quantile matching experiment induces posteriors that assign probability $q$ to $\theta$ and assign probability $1-q$ to $F^{-1}(q+ (1-q)F(\theta)/q)$, for $\theta \in [0,F^{-1}(q)]$.} Formally, $\tau^{\star}$ is defined by
\[
\tau^{\star}(M)=\int_0^q \1 \{q\delta_{F^{-1}(\omega)}+(1-q)\delta_{F^{-1}(q+\frac{1-q}{q}\omega)}\in M\}\frac{\df \omega}{q},\quad \text{for all $M\subset \Delta(\Theta)$}.
\]

While all of our results hold for general $F$ and $q$, we will provide intuition only for the uniform-median case where $F$ is uniform on $[0,1]$ and $q=1/2$. The uniform-median case is essentially without loss of generality, up to a change of variables, because the distribution function of $\theta=F^{-1}(\omega)$ is $F$ if $\omega$ is uniformly distributed on $[0,1]$.

A distribution $H$ of $q$-quantiles is \emph{implemented} by an experiment $\tau$ if there exists a (measurable) selection $\chi$ from $X$ such that the distribution of $\chi(G)$ induced by $\tau$ is $H$. A distribution $H$ of $q$-quantiles is \emph{uniquely implemented} by an experiment $\tau$ if $H$ is implemented by $\tau$ and  $X(G)$ is a singleton for $\tau$-almost all $G$. 
Let $\mathcal H$ and $\mathcal H^\star$ be the sets of implementable and uniquely implementable distributions of $q$-quantiles.

The following theorem characterizes $\mathcal H$ and $\mathcal H^\star$.

\begin{theorem}\label{t:implementable} The following hold:
\begin{enumerate}
\item $\mathcal H=\{H\in \Delta(\Theta):\ul H\leq H\leq \ol H\}$, where $\ul H(x)=\max \{0,{(F(x)-q)}/{(1-q)} \}$ and $\ol H(x)=\min \{{F(x)}/{q}, 1\}$ for all $x\in \Theta$. 
\item Every $H\in \mathcal H$ is implemented by $\tau^\star$.
\item If $F$ has a positive density on $\Theta$ then $\mathcal H$ is the closure of $\mathcal H^\star$. In particular, for any objective function $V\in C(\Theta)$, we have
\begin{equation}\label{e:sup=max}
\sup_{H\in \mathcal H^\star} \int_\Theta V(x)\df H(x)=\max_{H\in \mathcal H} \int_\Theta V(x)\df H(x).
\end{equation}
\end{enumerate}
\end{theorem}
Figure 1 illustrates the set $\mathcal H$. The intuition for Theorem \ref{t:implementable} is straightforward.
First, by simple Markov-type inequalities, any implementable $H$ must satisfy $\ul H \leq H\leq \ol H$. For example, if the posterior median is less than $x$ with probability $p$, then $\theta$ must be less than $x$ with probability at least $p/2$. When $F(x)=x$, this implies that $p \leq 2x$, so the probability that the posterior median is less than $x$ is at most $\min \{2x,1\}= \ol H(x)$.\footnote{This argument is closely related to \citeauthor{KG}'s \citeyear{KG} ``prosecutor-judge'' example. As in their example, the key observation is that if the prior probability of an event (e.g., the event that $\theta\leq x$) is $x$, then the maximum probability that the posterior probability of this event is at least $1/2$ is $\min \{2x,1\}$.
}

Conversely, to see that any $H$ satisfying $\ul H \leq H\leq \ol H$ is implementable, consider the median-matching experiment $\tau^{\star}$ that induces only posteriors $G_\theta$ that assign equal probability to some $\theta \in [0,1/2]$ and to $1/2 + \theta$. The set of medians of such a posterior is $X(G_{\theta})= [\theta, 1/2+ \theta]$. At the same time, $H\leq \ol H$ implies that $H^{-1}(2 \theta)\geq \theta$, and $H\geq \ul H$ implies that $H^{-1}(2 \theta)\leq 1/2+\theta$, so we have $H^{-1}(2 \theta) \in [\theta, 1/2 +\theta]$. Thus, $\chi (G_{\theta})=H^{-1}(2 \theta)$ is a selection from $X(G_{\theta})$. Finally, the distribution of $\chi (G_{\theta})$ induced by $\tau^{\star}$ is $H$, because the states that induce medians below $x$ under $\tau^{\star}$ with selection $\chi (G_{\theta})$ are precisely those in $[0,H(x)/2]$ and $[1/2,1/2+H(x)/2]$, and the measure of these states is $H(x)$.

Finally, for unique implementation, for any $e\in (0,1]$ and any implementable and absolutely continuous $H$ with density $h$, we can explicitly construct a modification of the median-matching experiment $\tau^{\star}_e$ that uniquely implements the distribution $(1-e)H+e F$ of medians, by making every posterior $G_\theta$ a convex combination of the median matching distribution $(\delta_{\theta}+\delta_{1/2+\theta})/2$ and the degenerate distribution $\delta_{H^{-1}(2\theta)}$ at the unique median $H^{-1}(2\theta)\in [\theta,1/2+\theta]$. 
Intuitively, for each $\theta\in [0,1/2]$, $\tau^\star_e$ induces posteriors $G_{\theta}$ and $G_{H(\theta)/2}$ with probabilities $1-e$ and $e$; similarly, for each $\theta\in (1/2,1]$, $\tau^\star_e$ induces posteriors $G_{\theta-1/2}$ and $G_{H(\theta)/2}$ with probabilities $1-e$ and $e$. Then posterior medians in $[x,x+\df x]$ are induced at $\theta\in [{H(x)}/{2},{H(x+\df x)}/{2}]$ with probability $1-e$, at $\theta\in [1/2 +H(x)/2,{1/2+H(x+\df x)/2}]$ with probability $1-e$, and at $\theta \in [x,x+\df x]$ with probability $e$. Since $H(x+\df x)=H(x)+h(x)\df x$, the density of the posterior median $x$ multiplied by the posterior at $x$ is equal to $(1-e)h(x)({\delta_{{H(x)}/{2}}}+{\delta_{{1}/{2}+{H(x)}/{2}}})/2+e\delta_x$,
as required. 
To complete the proof of Theorem \ref{t:implementable}, we provide a simple argument showing that any distribution in $\mathcal H$ can be approximated by uniquely implementable distributions $(1-e)H+e F$.\footnote{A complete characterization of the set $\mathcal H^\star$ remains an open problem. Two observations are that $\mathcal H^\star$ is a proper subset of $\mathcal H$ (as $\ul H$ and $\ol H$ do not belong to $\mathcal H^\star$), and that not all uniquely implementable distributions can be implemented by our modification of $q$-quantile matching. For example, $H=\delta_{1/2}$ is uniquely implemented by complete pooling but not by our modification of median matching.}

The literature contains several close antecedents of Theorem \ref{t:implementable}. \citet{FH} study partisan gerrymandering with a finite number of legislative districts. \citet{BD} study testing for overconfidence with a finite number of bins in a self-ranking experiment. In our notation, \citeauthor{FH} and \citeauthor{BD} consider discrete experiments with finitely many induced posteriors. \citeauthor{FH} show that a discrete version of $\ul H$ is the highest implementable distribution of posterior medians. \citeauthor{BD} show that the set of uniquely implementable distributions of posterior medians is a discrete version of the set $\{H\in \Delta(\Theta):\ul H< H< \ol H\}$. In a general setting with possibly infinitely many induced posteriors in the contexts of quantile persuasion and partisan gerrymandering, respectively, \citet{KW20} and \citet{KW} show that $\ul H$ is the highest implementable distribution of posterior medians. Finally, in a general setting, \citet{YangZentefis} show that the set of implementable distributions of posterior medians is $\{H\in \Delta(\Theta):\ul H\leq H\leq \ol H\}$, and also construct a dense subset of distributions that are uniquely implementable.\footnote{To establish results similar to our Theorem \ref{t:implementable}, \citeauthor{YangZentefis} characterize the extreme points of the set $\{H\in \Delta(\Theta):\ul H\leq H\leq \ol H\}$. As recently emphasized by \citet{KMS}, characterizing a convex set by its extreme points can be useful for establishing a given property of the set. In contrast, we show that directly characterizing the set of implementable distributions of posterior quantiles is much easier than characterizing the extreme points of this set.} Relative to \citeauthor{BD} and \citeauthor{YangZentefis}, Theorem \ref{t:implementable} shows that the $q$-quantile matching experiment implements every $H\in \mathcal H$, and also yields a much simpler proof.  


\section{Optimal Distributions of Posterior Quantiles}
\label{s:optimal}

This section uses the $q$-quantile matching experiment to characterize the distributions of posterior $q$-quantiles that optimize a continuous linear functional.

\begin{theorem}\label{t:opt}
Let $V\in C(\Theta)$. Then $H$ (uniquely) maximizes $\int V(x)\df H(x)$ on $\mathcal H$ iff $H^{-1}(p)$ (uniquely) maximizes $V$ on $[\ol H^{-1}(p),\ul H^{-1}(p)]$ for (almost) all $p\in [0,1]$. Consequently, the value of the maximization problem is
\begin{equation}\label{e:opt}
\max_{H\in \mathcal H} \int_\Theta V(x)\df H(x)=
\int_0^1 \max\{V(x):\, x\in [\ol H^{-1}(p),\ul H^{-1}(p)]\}\df p.	
\end{equation}
\end{theorem}
Conceptually, Theorem \ref{t:opt} follows easily from Theorem \ref{t:implementable}.  
Since the median matching experiment $\tau^{\star}$ implements all implementable distributions of medians, optimization just requires selecting an optimal median $\chi (G_\theta) \in \argmax_{x\in [\theta, 1/2+\theta]}V(x)$ for each posterior $G_\theta$ induced by $\tau^{\star}$, as illustrated in Figure 1.
The value of the maximization problem is thus $2\int_{0}^{1/2} \max_{x\in [\theta, 1/2+\theta]} V(x)\df \theta$, and a distribution $H$ of medians is optimal iff $H^{-1}(2\theta) \in \argmax_{x\in [\theta, 1/2+\theta]}V(x)$ for all $\theta \in [0,1/2]$. That is, optimal solutions can be obtained by pointwise maximization without any ironing procedure.  

In general, by Theorem \ref{t:implementable}, for each $H\in \mathcal H$ and $p\in [0,1]$, we have $\ol H^{-1}(p)\leq H^{-1}(p) \leq \ul H^{-1}(p)$. If we consider the relaxed problem of finding a measurable function $J:[0,1] \rightarrow \Theta $ to
\begin{gather*}
\text{maximize }\int_0^1 V(J(p))\df p\\
\text{subject to }\ol H^{-1}(p)\leq J(p) \leq \ul H^{-1}(p),\quad \text{for all }p\in [0,1],
\end{gather*}
one solution is
\[
J^\star(p)=\min \argmax_{}\{V(x):\, x\in [\ol H^{-1}(p),\ul H^{-1}(p)]\}, \quad \text{for all }p\in [0,1].
\]
This function $J^\star$ is monotone; moreover, the proof of Theorem \ref{t:opt} shows that there exists $H^\star \in \Delta(\Theta)$ such that $J^\star = H^{\star -1}$, so $H^\star$ solves the optimization problem  \eqref{e:opt}.

The closest antecedent to Theorem \ref{t:opt} is Corollary 4 of \citet{YangZentefis}, which solves the maximization problem \eqref{e:opt} in the special cases where $V$ is quasi-concave or quasi-convex. The solution follows immediately from Theorem \ref{t:opt}. 
To see how, suppose that $V$ is quasi-concave with a maximum at $x^\star\in [0,1]$. For each interval $[\theta,1/2+\theta]$, it is optimal to select $x^\star$ if $x^\star \in [\theta,1/2+\theta]$, $\theta$ if $x^\star<\theta$, and $1/2+\theta$ if $x^\star>1/2+\theta$. This induces the distribution of posterior medians
\begin{align}\label{e:concave}
H(x)=
\begin{cases}
\ul H(x), &x<x^\star,\\
\ol H(x), &x\geq x^\star.
\end{cases}
\end{align}
Next, suppose that $V$ is quasi-convex with $V(x^\star)=V(1/2+x^\star)$ for some $x^\star \in  [0,1/2]$. Then, for each interval $[\theta,1/2+\theta]$, it is optimal to select $\theta$ if $x^\star>\theta$ and $1/2+\theta$ if $x^\star <\theta$. This induces the distribution of posterior medians
\begin{align}\label{e:convex}
H(x)=
\begin{cases}
\ol H(x), &x<x^\star,\\
2x^\star, &x\in [x^\star,\frac 12+x^\star),\\
\ul H(x), &x\geq \frac 12+x^\star.
\end{cases}
\end{align}

From the perspective of optimization, it is natural to ask whether each extreme point of $\mathcal H$ is exposed, meaning that it is the unique maximizer in $\mathcal H$ of $\int V(x)\df H(x)$ for some $V\in C(\Theta)$. It turns out that some extreme points are not exposed. 
To see this, note that in the uniform-median case the distribution $H^\star=(\delta_{1/4}+\delta_{1/2})/2$ is an extreme point of $\mathcal H$, as there are no distinct $H_1,H_2\in \mathcal H$ such that $H^\star=(H_1+H_2)/2$. By Theorem \ref{t:opt}, if $H^\star$ maximizes $\int V(x)\df H(x)$ on $\mathcal H$ for some $V\in C(\Theta)$, then $V(1/4)\geq V(x)$ for all $x\in [\theta,1/2+\theta]$ and all $\theta\in [0, 1/4]$, and, similarly, $V(1/2)\geq V(x)$ for all $x\in [\theta,1/2+\theta]$ and all $\theta\in [0, 1/2]$. Thus, $V(1/4)=V(1/2)\geq V(x)$ for all $x\in [0,1]$. But then the distribution $\delta_{1/2}\in \mathcal H$ also maximizes $\int V(x)\df H(x)$, which shows that $H^\star$ is not an exposed point of $\mathcal H$.\footnote{The distribution $H^ \star$ does uniquely maximize $\int V (x)\df H(x)$ for $V = 2\cdot\1\{x = 1/4\}+\1 \{x = 1/2\}$, which is upper semi-continuous but not continuous. An open question is whether each extreme point of $\mathcal H$, characterized in Theorem 1 of \citet{YangZentefis}, is the unique maximizer of $\int V (x)\df H(x)$ for some upper-semicontinuous $V$. This is a weaker property than exposedness, as the usual theory of exposed points relies on continuity.}

\section{Unique Properties of the Quantile Matching Experiment}
\label{s:unique}

Theorem \ref{t:implementable} shows that the $q$-quantile matching experiment simultaneously implements all implementable distributions of posterior $q$-quantiles. We now show that it is the unique experiment to do so. For simplicity, in this section we assume that $F$ has a positive density on $\Theta$.

We actually establish the stronger result that the $q$-quantile matching experiment is the unique experiment that simultaneously implements all optimal distributions for strictly quasi-convex objective functions: that is, all distributions of the form of equation \eqref{e:convex}.

\begin{theorem}\label{t:uniqueness}
The $q$-quantile matching experiment $\tau^\star$ is the unique experiment $\tau$ that, for each $p \in [0,1]$, implements the distribution $H_p\in \mathcal H$ given by
\begin{equation}
\begin{gathered}
H_p(x)=
\begin{cases}
\ol H(x), &x<\ul x_p,\\
p, & x\in [\ul x_p,\ol x_p),\\
\ul H(x), & x\geq \ol x_p,
\end{cases}\\
\text{where}\ \ \ul x_p=F^{-1}(qp)\ \ \text{and}\ \  \ol x_p=F^{-1}(q+(1-q)p).
\end{gathered}
\end{equation}
\end{theorem}

In other words, for any experiment $\tau \neq \tau^\star$, there is some $p \in [0,1]$ such that $\tau$ does not implement $H_p$. For example, in the uniform-median case, the negative assortative matching experiment does not implement $H_{1/2}$, as noted in the introduction.

An immediate corollary of Theorem \ref{t:uniqueness} is that the $q$-quantile matching experiment is the unique experiment that minimizes the maximum regret of a designer who chooses an experiment $\tau$ before learning her objective $V$, but chooses a selection $\chi$ after learning $V$. Formally, for each experiment $\tau\in \Delta(\Delta(\Theta))$ and each objective $V\in  C(\Theta)$ define the designer's \emph{regret} as
\[
r (\tau,V)=\max_{H\in \mathcal H}\int_\Theta V(x)\df H(x)-\sup_{H\in \mathcal H}\left \{\int_\Theta V(x)\df H(x):\, \text{$H$ is implemented by $\tau$} \right \}.
\]
Note that $r (\tau,V) \geq 0$ for all $\tau$ and $V$. Say that a set of possible objective functions $\mathcal V\subset C(\Theta)$ is \emph{rich} if, for all $x_0,x_1\in \Theta$, there exists a strictly quasi-convex $V\in \mathcal V$ with $V(x_0)=V(x_1)$. We then have the following result.

\begin{corollary}\label{c:regret}
If $\mathcal V$ is rich then the $q$-quantile matching experiment $\tau^\star$ is the unique experiment $\tau$ such that $r (\tau,V)=0$ for all $V\in \mathcal V$.
\end{corollary}

\section{Discussion}
\label{s:discussion}

We conclude by discussing implications of our results in the gerrymandering context and mentioning some directions for future research.

A standard model of gerrymandering indexes voters by a one-dimensional partisan type $\theta \in \Theta= [\underline{\theta},\bar{\theta}]$ with population distribution $F$, and views a \emph{district} $G \in \Delta (\Theta)$ as a distribution over voter types. A \emph{districting plan} $\tau\in \Delta (\Delta (\Theta))$ is then a distribution over districts that satisfies the constraint $\int G \df \tau(G) =F$, reflecting the legal requirement that all voters must be assigned to equipopulous districts. The outcome of an election held under a districting plan then depends on the realization of an aggregate shock, whose realization $\rho$ is drawn from some distribution $R$ on $\R $. In particular, party 1 wins a district $G$ when the aggregate shock takes value $\rho$ iff $\int v(\theta,\rho) \df G(\theta) \geq 1/2$, where $v(\theta,\rho)$ is party 1's vote share among type-$\theta$ voters at aggregate shock $\rho$.\footnote{This formulation follows \citet{KW}, which generalizes the earlier models of \citet{OG}, \citet{FH}, and \citet{GP}. The connection between gerrymandering and distributions of posterior quantiles is also discussed by \citet{YangZentefis}.}

An important special case of this model arises when $v(\theta,\rho)=\mathbf{1}\{\theta \geq \rho \}$, so there is no residual uncertainty about individual voter behavior conditional on the aggregate shock. In this ``no individual uncertainty'' case (which approximates the model of \citealt{FH}), party 1 wins district $G$ at aggregate shock $\rho$ iff the median of $G$, $\chi(G)$, exceeds $\rho$. The distribution of the parties' seat shares under a districting plan $\tau$ is thus determined by the distribution of posterior medians under $\tau$: in particular, the designer wins each district $G$ with probability $R(\chi(G))$. Consequently, under \emph{partisan gerrymandering}, where the districting plan is designed to maximize party 1's expected seat share---so that the designer's utility from creating a district with a median voter of type $x$ is $V(x)=R(x)$---the resulting districting plan is $\underline{H}$, the highest implementable distribution of medians (\citealt{KW}, Proposition 3).

While partisan gerrymandering has attracted much attention, other practically relevant models include \emph{bipartisan gerrymandering}---where incumbent politicians of both parties design a districting plan to maximize the security of their own districts---and \emph{non-partisan gerrymandering}---where an independent districting commission designs a districting plan to make districts as competitive as possible.\footnote{On bipartisan gerrymandering, see, e.g., \citet{Issacharoff}, \citet{Persily}, and \citet{KonishiPan}. The objective of creating competitive districts has not been studied in the literature as far as we know, but it is an established districting criterion in several US states: see, e.g., https://www.ncsl.org/redistricting-and-census/redistricting-criteria.} The results of the current paper have implications for optimal bipartisan and non-partisan gerrymandering (without individual uncertainty). In particular, it seems natural that a bipartisan designer's utility for creating a district with a median voter of type $x$ should be quasi-convex in $x$, since districts with low or high medians are safer for the stronger party in that district: i.e.,  the stronger party's probability of winning a district with median $x$, $\max \{R(x),1-R(x)\}$, is quasi-convex in $x$. By Theorem \ref{t:opt}, optimal bipartisan districting plans are then given by \eqref{e:convex}. Conversely, a non-partisan designer's utility for creating a district with median $x$ may be quasi-concave in $x$, since districts with moderate medians are more competitive: i.e., the weaker party's probability of winning a district with median $x$, $\min \{R(x),1-R(x)\}$, is quasi-concave in $x$. Optimal non-partisan districting plans are then given by \eqref{e:concave}. More generally, Theorem \ref{t:implementable} characterizes all implementable distributions of the parties' seat shares; and Theorem \ref{t:uniqueness} shows that the only districting plan that simultaneously implements all such distributions is the median-matching plan that forms districts by pairing voter types across the median of $F$ in a positively assortative manner.

In the gerrymandering context, characterizing the set of implementable distributions of seat shares with both aggregate and individual uncertainty is an open question. More generally, this question corresponds to characterizing the set of implementable distributions of posterior statistics that interpolate between the median (the relevant statistic when $v(\theta,\rho)=\mathbf{1}\{\theta \geq \rho \}$) and the mean (the relevant statistic when $v(\theta,\rho)=\theta-\rho $). \citet{BD}, \citet{YangZentefis}, and the current paper characterize implementable distributions of posterior medians (or other quantiles); \citet{Blackwell}, \citet{Strassen}, and \citet{Kolotilin2017} characterize implementable distributions of posterior means. The problem of characterizing optimal distributions of intermediate statistics---the analogous problem to that of Theorem \ref{t:opt} in the current paper---is studied in the gerrymandering context by \citet{KW}, and in general by \citet{KCW23}. It is currently unknown whether a useful analog of Theorem \ref{t:implementable} (for quantiles) or Strassen's theorem (for means) exists for intermediate statistics.

\section{Proofs}
\label{s:proofs}

\begin{proof}[Proof of Theorem \ref{t:implementable}]
Consider any experiment $\tau\in \Delta(\Delta(\Theta))$ and any measurable selection $\chi(G)$ from $X(G)$. Let $H$ be the distribution of $\chi(G)$ induced by $\tau$. Then, for each $x\in \Theta$, we have
\begin{gather*}
	F(x)=\int G(x)\df \tau(G)=\int \1\{G(x)\geq q\}G(x)\df \tau(G)+\int \1 \{G(x)<q\}G(x)\df \tau (G)\\
	\geq \int \1\{G(x)\geq q\}q\df \tau(G)\geq \int \1 \{\chi (G)\leq x\}q\df \tau(G)=qH(x),
\end{gather*}
showing that $H\leq \ol H$. A symmetric argument shows that $H \geq \ul H$.

For the converse, we first note that the median-matching experiment $\tau^{\star}$ is  well-defined because $\int G \df \tau^{\star}(G)=F$: indeed, for all  $\theta\in \Theta$, we have
\begin{align*}
\int G(\theta)\df \tau^{\star}(G)&=\int_0^q \left (q\delta_{F^{-1}(\omega)}+(1-q)\delta_{F^{-1}(q+\frac{1-q}{q}\omega)}\right)\left(\theta\right)\frac{\df \omega}{q} \\
&=
\begin{cases}
\int_0^{F(\theta)}q\frac{\df \omega}{q}, &F(\theta)< q,\\
\int_0^{q}q\frac{\df \omega}{q} +\int_0^{\frac{q}{1-q}(F(\theta)-q)}(1-q)\frac{\df \omega}{q}, &F(\theta)\geq q,
\end{cases}\: = F(\theta),
\end{align*}
where the second equality holds because $F^{-1}(\omega)\leq \theta$ iff $\omega\leq F(\theta)$, and $F^{-1}(q+\frac{1-q}{q}\omega)\leq \theta$ iff $\omega\leq \frac{q}{1-q}(F(\theta)-q)$. Note also that, for each $\omega\in [0,q]$, the set of $q$-quantiles of $G_\omega$ is $X(G_\omega)= [F^{-1}(\omega), F^{-1}(q+ \frac{1-q}{q}\omega)]$.

Now fix a distribution $H\in \Delta(\Theta)$ satisfying $\ul H\leq H\leq \ol H$. Note that, for each $\omega\in [0,q]$, since $H\leq \ol H$, we have $H^{-1}(\frac \omega q)\geq \ol H^{-1}(\frac \omega q)=F^{-1}(\omega)$; and, since $H\geq \ul H$, we have $H^{-1}(\frac \omega q)\leq \ul H^{-1}(\frac \omega q)=F^{-1}(q+\frac{1-q}{q}\omega)$. Thus, $H^{-1}(\frac \omega q) \in X(G_\omega)$. We can therefore define a selection $\chi(G)$ from $X(G)$ by letting $\chi(G)=H^{-1}(\frac \omega q)$ in the $\tau^{\star}$-almost sure event that $G=G_\omega$ for some $\omega\in [0,q]$, and (for concreteness) letting $\chi(G)=\min X(G)$ otherwise. Finally, the distribution of $\chi(G)$ induced by $\tau^{\star}$ is $H$, because, for all $x \in \Theta$, we have
\begin{align*} 
\int \1 \{\chi(G)\leq x\}\df \tau^{\star}(G)=\int \1 \{H^{-1}(\tfrac \omega q)\leq x\}\frac{\df \omega}{q} =\int_0^{qH(x)}\frac{\df \omega}{q} = H(x).
\end{align*}

For unique implementation, assume that $F$ has a positive density on $\Theta=[\ul \theta,\ol \theta]$. Fix any $H\in \mathcal H$.  Consider a sequence of partitions of $\Theta$ given by $\theta_{i,n}=\ul \theta +(\ol \theta-\ul \theta)\frac {i}{2^n}$, with $i\in \{0,1,\dots, 2^n\}.$ Define a sequence $H_n\in \Delta (\Theta)$ by
\[
H_n(x)=H(\theta_{i-1,n})\frac{F(\theta_{i,n})-F(x)}{F(\theta_{i,n})-F(\theta_{i-1,n})}+H(\theta_{i,n})\frac{F(x)-F(\theta_{i-1,n})}{F(\theta_{i,n})-F(\theta_{i-1,n})},
\]
for all $i\in \{1,\dots,2^n\}$ and all $x\in [\theta_{i-1,n},\theta_{i,n}]$. Note that $H_n$ is well-defined, because $F$ is strictly increasing on $\Theta$. Since $H\in \mathcal H$, we have $H_n\in \mathcal H$. Moreover, $H_n$ has a simple density function $h_n$ with respect to $F$, given by
\[
h_n(x)=\frac{H(\theta_{i,n})-H(\theta_{i-1,n})}{F(\theta_{i,n})-F(\theta_{i-1,n})},
\]
for all $i\in \{1,\dots,2^n\}$ and all $x\in (\theta_{i-1,n},\theta_{i,n})$.

Next, for each $e\in (0,1]$ and each $n$, there exists an experiment $\tau^{\star}_{e,n}\in \Delta(\Delta(\Theta))$ satisfying the following two properties. First, for $\tau^{\star}_{e,n}$-almost all $G$, there exists $x\in \Theta$ such that $G=G^x$ where
\[
G^x=\frac{(1-e)h_n(x)(q\delta_{F^{-1}(qH_n(x))}+(1-q)\delta_{F^{-1}(q+(1-q)H_n(x))})+e \delta_{x}}{(1-e)h_n(x)+e}.
\]
This implies that $X(G^x)$ is the singleton $\{x\}$, because $e>0$ and $F^{-1}(qH_n(x))\leq x\leq F^{-1}(q+(1-q)H_n(x))$ (which holds because $\ul H(x)\leq H_n(x)\leq \ol H(x)$). Second, the distribution of unique quantiles $\chi(G^x)=x$ induced by $\tau^{\star}_{e,n}$ is $(1-e)H_n+eF$.

Formally, $\tau^{\star}_{e,n}$ is defined by
\begin{gather*}
\tau^{\star}_{e,n}(M)=\int_0^1 \1\left\{G^x\in M\right\}\left((1-e)h_n(x)+e\right)\df F(x), \quad \text{for all $M\subset \Delta(\Theta)$}.
\end{gather*}
Note that $\tau^{\star}_{e,n}$ is a well-defined experiment because $\int G \df \tau^{\star}_{e,n}(G)=F$. Indeed, for all  $\theta\in \Theta$, we have
\begin{align*}
\int G(\theta)\df \tau^{\star}_{e,n}(G)
&=\int_0^1 ((1-e)h_n(x)(q\delta_{F^{-1}(qH_n(x))}+(1-q)\delta_{F^{-1}(q+(1-q)H_n(x))})+e\delta_x)\left(\theta\right)\df F(x) \\
&=
(1-e)\begin{cases}
\int_0^{H_n^{-1}(\frac {F(\theta)} {q})}q\df H_n(x), &F(\theta)< q,\\
\int_0^{H_n^{-1}(1)}q\df H_n(x) +\int_0^{H_n^{-1}(\frac{F(\theta)-q}{1-q})}(1-q)\df H_n(x), &F(\theta)\geq q,
\end{cases}\\
&\quad +e\int_0^\theta \df F(x)
=F(\theta),
\end{align*}
where the second equality holds because  $F^{-1}(qH_n(x))\leq \theta$ iff $x\leq H_n^{-1}(\frac{F(\theta)}{q})$, and $F^{-1}(q+(1-q)H_n(x))\leq \theta$ iff $x\leq H_n^{-1}(\frac{F(\theta)-q}{1-q})$; and the third equality holds because $H_n(H_n^{-1}(p))=p$ for all $p\in [0,1]$, by continuity of $H_n$ (which holds because $H_n$ has a density with respect to $F$ and $F$ has a density with respect to the Lebesgue measure).
Finally, the distribution of unique quantiles $\chi(G)$ induced by $\tau^{\star}_{e,n}$ is $(1-e)H_n+eF$, because, for all $y\in \Theta$, we have
\begin{align*}
\int \1 \{\chi(G)\leq y\}\df \tau^{\star}_{e,n}(G)=\int_0^{y} ((1-e)h_n(x)+e)\df F(x) = (1-e)H_n(y)+eF(y).
\end{align*}

Now fix $V\in C(\Theta)$. By continuity of $V$ and compactness of $\Theta$, for all $\varepsilon >0$, there exists $N\in \mathbb N$ such that (i) $|V(x)-V(y)|\leq \varepsilon$ for all $x,y\in [\theta_{i-1,n},\theta_{i,n}]$, all $i\in \{1,\dots.,2^n\}$, and all $n\geq N$, and (ii) $e|V(x)-V(y)|\leq \varepsilon$ for all $x,y\in \Theta$ and all $e\in (0,\frac{1}{N}]$. Then, for all $n\geq N$ and $e\in (0,\frac{1}{N}]$, we have
\begin{gather*}
\left|\int V(x)\df H(x)-\int V(x)\df ((1-e)H_n+eF)(x)\right|\\
\leq (1-e)\left|\int V(x)\df (H-H_n)(x)\right| +e\left|\int V(x)\df (H-F)(x)\right|\leq \varepsilon+\varepsilon.
\end{gather*}
Since this holds for any $V\in C(\Theta)$, it follows that $(1-\frac{1}{n})H_n+\frac{1}{n} F$ converges weakly to $H$. In turn, since we have seen that the experiment $\tau^{\star}_{1/n,n}$ uniquely implements $(1-\frac{1}{n})H_n+\frac{1}{n} F$, we conclude that $\mathcal H^\star$ is dense in $\mathcal H$. Finally, $\mathcal H$ is compact, as $\Delta(\Theta)$ is compact by Theorem 15.11 in \citet{aliprantis2006}, and $\mathcal H$ is the intersection over $x\in \Theta$ of the closed subsets $\mathcal H_x:=\{H\in \Delta(\Theta): \ul H(x)\leq H(x)\leq \ol H(x)\}$ of $\Delta(\Theta)$. Thus, the closure of $\mathcal H^\star$ is $\mathcal H$, and hence \eqref{e:sup=max} holds for any $V\in C(\Theta)$. 
\end{proof}

\begin{proof}[Proof of Theorem \ref{t:opt}]
For each $H\in \Delta(\Theta)$, we have
\[
\int_\Theta V(x)\df H(x)= \int_0^1 V(H^{-1}(p))\df p.
\]
Recall that 
\[
J^\star(p)=\min \argmax_{}\{V(x):\, x\in [\ol H^{-1}(p),\ul H^{-1}(p)]\}, \quad \text{for all }p\in [0,1].
\]
Since $J^\star$ is defined as the minimum selection from the $\argmax$, it follows that (i) $J^\star$ is non-decreasing (and hence measurable), because $\ol H^{-1}$ and $\ul H^{-1}$ are non-decreasing, (ii) $J^\star$ is left-continuous, because $\ol H^{-1}$ and $\ul H^{-1}$ are left-continuous and $V\in C(\Theta)$, (iii) $J^\star(1)\leq \ol \theta$, because $\ul H^{-1}(1)\leq \ol \theta$, and (iv) $J^\star(0)= \ul \theta$, because $\ol H^{-1}(0)=\ul H^{-1}(0)= \ul \theta$.
This implies that $J^\star=H^{\star-1}$, where $H^{\star}\in \Delta(\Theta)$ is given by
\[
H^\star(x) =\sup\{p\in [0,1]:\, J^\star(p)\leq x\}, \quad \text{for all }x\in \Theta.
\]
Moreover, since $\ol H^{-1}\leq J^\star\leq \ul H^{-1}$, it follows that $H^\star \in \mathcal H$, so $H^\star $ solves the original problem, and its value coincides with the value of the relaxed problem, yielding \eqref{e:opt}. Consequently, $H\in \mathcal H$ maximizes $\int V(x)\df H(x)$ on $\mathcal H$ iff $H^{-1}(p)$ maximizes $V$ on $[\ol H^{-1}(p),\ul H^{-1}(p)]$ for almost all $p\in [0,1]$. Moreover, by continuity of $V$ and left-continuity of $\ol H^{-1}$ and $\ul H^{-1}$, $H^{-1}(p)$ maximizes $V$ on $[\ol H^{-1}(p),\ul H^{-1}(p)]$ for almost all $p\in [0,1]$ iff it does so for all $p\in [0,1]$.

Next, if $V$ has a unique maximum on $[\ol H^{-1}(p),\ul H^{-1}(p)]$ for almost all $p\in [0,1]$, then $J^\star$ is the unique solution of the relaxed problem that satisfies properties (i)--(iv), and hence $H^\star$ is the unique solution of the original problem. Conversely, if there exists a non-negligible set $P\subset [0,1]$ such that $V$ has multiple maxima on $[\ol H^{-1}(p),\ul H^{-1}(p)]$ for each $p\in P$, then there are multiple solutions of the relaxed problem that satisfy properties (i)--(iv). For example, $\hat J$ defined as the maximum selection from the $\argmax$
 also solves the relaxed problem, and so does $\hat J^{\star}$ defined by $\hat J^{\star}(p)=\hat J(p^-)$ for all $p\in (0,1]$ and $\hat J^{\star}(0)=\ul \theta$. But, by construction, $\hat J^{\star}$ satisfies properties (i)--(iv) and is not equal to $J^{\star}$. Then $\hat J^{\star}=\hat H^{\star-1}$ where $\hat H^{\star}\in \Delta(\Theta)$  is given by
\[
\hat H^{\star}(x) =\sup\{p\in [0,1]:\, \hat J^{\star}(p)\leq x\}, \quad \text{for all }x\in \Theta.
\]
Thus, $\hat H^{\star}\neq H^\star$ also solves the original problem.
\end{proof}

\begin{proof}[Proof of Theorem \ref{t:uniqueness}]
Suppose that an experiment $\tau\in \Delta(\Delta (\Theta))$ implements all $H_p$. Fix any $p\in [0,1]$.
Since $\tau$ implements $H_p$, there exists a measurable selection $\chi_p(G)$ from $X(G)$ such that the distribution of $\chi_p(G)$ induced by $\tau$ is $H_p$. Since $F$ has a density on $\Theta$, we have
\begin{gather*}
qp=	F(\ul x_p)=\int G(\ul x_p)\df \tau(G)=\int \1\{G(\ul x_p)\geq q\}G(\ul x_p)\df \tau(G)\\
+\int \1 \{G(\ul x_p)<q\}G(\ul x_p)\df \tau (G)\geq \int \1\{G(\ul x_p)\geq q\}q\df \tau(G)\\
	\geq \int \1 \{\chi_p (G)\leq \ul x_p\}q\df \tau(G)
	=qH_p(\ul x_p)=qp,
\end{gather*}
so all inequalities hold with equality. Thus, $\tau(G(\ul x_p)=0)=1-p$, $\tau(G(\ul x_p)=q)=p$, and $\tau(\chi_p(G)\leq \ul x_p)=p$.  A symmetric argument yields $\tau(G(\ol x_p)=q)=1-p$, $\tau(G(\ol x_p)=1)=p$, and $\tau (\chi_p(G)>\ol x_p)=1-p$. Next, since $G(\ul x_p)=0$ and $G(\ol x_p)=1$ imply that $\ul x_p<\chi_p(G)\leq \ol x_p$, it follows that $\tau(G(\ul x_p)=0,\, G(\ol x_p)=1)=0$, because
\[
\tau(\ul x_p<\chi_p(G)\leq \ol x_p)=1-\tau(\chi_p(G)\leq \ul x_p)-\tau (\chi_p(G)>\ol x_p)=1-p-(1-p)=0.
\]
So, $\tau(G(\ul x_p)=0,\, G(\ol x_p)=q)=\tau(G(\ul x_p)=0)-\tau(G(\ul x_p)=0,\, G(\ol x_p)=1)=1-p$ and $\tau(G(\ul x_p)=q,\, G(\ol x_p)=1)=\tau(G(\ol x_p)=1)-\tau(G(\ul x_p)=0,\, G(\ol x_p)=1)=p$.
In sum,
\begin{equation}\label{e:tau}
\begin{aligned}
\tau(G(\ul x_p)=0,\, G(\ol x_p)=q)&=1-p,  \\
\tau(G(\ul x_p)=q,\, G(\ol x_p)=1)&=p,	
\end{aligned}
\qquad\text{for all $p\in [0,1]$.}
\end{equation}
We now show that \eqref{e:tau} yields $\tau =\tau^\star$. Let $X_0=[\ul x_0,\ol x_0]=[\ul \theta, F^{-1}(q)]$ and $X_1=[\ul x_1,\ol x_1]=[F^{-1}(q),\ol \theta]$. For each experiment $\tilde \tau\in \Delta(\Delta(\Theta))$, define a joint distribution function $I_{\tilde \tau}:X_0\times X_1\to [0,1]$ by
\[
I_{\tilde \tau} (x_0,x_1)=\tilde \tau (G=q\delta_{\theta_0}+(1-q)\delta_{\theta_1},\, \theta_0\in [\ul x_0,x_0],\,  \theta_1\in [\ul x_1,x_1]).
\]
To prove that $\tau =\tau^\star$, it suffices to show that $I_{\tau} (x_0,x_1)=I_{\tau^\star} (x_0,x_1)$ for all $(x_0,x_1)\in X_0\times X_1$, with $I_{\tau} (\ol x_0,\ol x_1)=I_{\tau^\star} (\ol x_0,\ol x_1)=1$. Fix any $(x_0,x_1)\in X_0\times X_1$, and let $\hat p= \tfrac{F(x_0)}{q}\wedge \tfrac{F(x_1)-q}{1-q}$. First, by definition of $\tau^\star$, we have
\[
I_{\tau^\star} (x_0,x_1)=\int_0^q \1 \{F^{-1}(\omega)\leq x_0,\, F^{-1}(q+\tfrac{1-q}{q}\omega)\leq x_1\}\frac{\df \omega}{q} =\hat p,
\]
with $I_{\tau^\star} (\ol x_0,\ol x_1)=1$, because $F(\ol x_0)=F(F^{-1}(q))=q$ and $F(\ol x_1)=F(\ol \theta)=1$.
Second, by \eqref{e:tau} and definition of $I_\tau$, we have
\begin{gather*}
\hat p=\tau(G(\ul x_{\hat p})=q,\, G(\ol x_{\hat p})=1)\leq I_\tau(x_0,x_1)\leq \tau((G(\ul x_{\hat p}),G(\ol x_{\hat p}))\neq (0,q))=1-(1-\hat p),
\end{gather*}
showing that $I_\tau(x_0,x_1)=\hat p=I_{\tau^\star} (x_0,x_1)$.
\end{proof}

\begin{proof}[Proof of Corollary \ref{c:regret}]
Consider any experiment $\tau \neq \tau^\star$. By Theorem \ref{t:uniqueness}, there exists $p\in [0,1]$ such that $\tau$ does not implement $H_p$. Since $\mathcal V$ is rich, there exists a continuous and strictly quasi-convex $V\in \mathcal V$ with $V(\ul x_p)=V(\ol x_p)$. Then $\ul x_{\tilde p}$ uniquely maximizes $V$ on $[\ul x_{\tilde p},\ol x_{\tilde p}]$  for all $\tilde p\in [0,p)$, and $\ol x_{\tilde p}$ uniquely maximizes $V$ on $[\ul x_{\tilde p},\ol x_{\tilde p}]$  for all $\tilde p\in (p,1]$. By Theorem \ref{t:opt}, $H_p$ uniquely maximizes  $\int V(x)\df H(x)$ on $\mathcal H$. 

Suppose for contradiction that $r (\tau,V)=0$. Then there exists a sequence $H^n\in \mathcal H$ implemented by $\tau$ such that $\int V(x)\df H^n(x)\to \int V(x)\df H_p(x)$.
Since $\mathcal H$ is weak$^\star$ compact, passing to a subsequence if necessary, we can assume that $H^n\to\hat  H\in \mathcal H$.  Note that $\hat H=H_p$, because $H_p$ uniquely maximizes $\int V(x)\df H(x)$ on $\mathcal H$. In sum, there exists a sequence of measurable selections $\chi^n(G)$ from $X(G)$ such that $H^n(x)=\tau(\chi^n(G)\leq x)$ and $H^n(x)\to H_p(x)$ for all $x\in \Theta$.

Next, we show that $\tau$ cannot simultaneously satisfy the following three conditions
\begin{align}
\tau(G(\ul x_{\tilde p})=0)=1-\tilde p\quad &\text{and}\quad \tau(G(\ul x_{\tilde p})=q)=\tilde p,\quad \text{for all $\tilde p\in [0,p]$},\label{e:1}\\
\tau(G(\ol x_{\tilde p})=q)=1-\tilde p\quad &\text{and}\quad \tau(G(\ol x_{\tilde p})=1)=\tilde p,\quad \text{for all $\tilde p\in [p,1]$},\label{e:2}\\
\tau(G(\ul x_p)=0,\, G(\ol x_p)=q)=1-p\quad &\text{and}\quad \tau(G(\ul x_p)=q,\, G(\ol x_p)=1)=p,\label{e:3}
\end{align}
as otherwise $\tau $ would implement $H_p$. 
Indeed, if $\tau$ satisfies \eqref{e:1}--\eqref{e:3}, we can define a selection $\chi(G)$ from $X(G)$ by letting $\chi(G)=\theta_0$ if $\theta_0<\ul x_p$ and $\chi(G)=\theta_1$ if $\theta_0>\ul x_p$ in the $\tau$-almost sure event that $G=q\delta_{\theta_0} +(1-q)\delta_{\theta_1}$ for some $\theta_0\in [\ul x_0,\ol x_0]$ and $\theta_1\in [\ul x_1,\ol x_1]$. Then the distribution of $\chi(G)$ induced by $\tau$ is $H_p$, because, by \eqref{e:1}, for all $\tilde p\leq p$, we have
\begin{gather*}
\tau(\chi(G)\leq \ul x_{\tilde p}) =\tau(G(\ul x_{\tilde p})=q)=\tilde p =\tfrac{F(\ul x_{\tilde p})}{q}=H_p(\ul x_{\tilde p}),
\end{gather*}
and, by \eqref{e:2} and \eqref{e:3}, for all $\tilde p\geq p$, we have
\begin{gather*}
\tau(\chi(G)\leq \ol x_{\tilde p}) =\tau(G(\ol x_{\tilde p})=1)+\tau(G(\ul x_p)=q,\, G(\ol x_{\tilde p})=q)=\tilde p =\tfrac{F(\ol x_{\tilde p})-q}{1-q}=H_p(\ol x_{\tilde p}).
\end{gather*}

Finally, we show that if at least one of conditions \eqref{e:1}--\eqref{e:3} fails, then $H^n\nrightarrow H_p$. First, if \eqref{e:1} fails at some $\tilde p\in [0,p]$, then there exists $\varepsilon>0$ such that
\begin{align*}
F(\ul x_{\tilde p})&=\int G(\ul x_{\tilde p})\df \tau(G)=\int \1\{G(\ul x_{\tilde p})\geq q\}G(\ul x_{\tilde p})\df \tau(G)+\int \1 \{G(\ul x_{\tilde p})<q\}G(\ul x_{\tilde p})\df \tau (G)\\
	&\geq\varepsilon +\int \1\{G(\ul x_{\tilde p})\geq q\}q\df \tau(G) \geq \varepsilon +\int \1 \{\chi^n (G)\leq \ul x_{\tilde p}\}q\df \tau(G) =\varepsilon +qH^n(\ul x_{\tilde p}),
\end{align*}
so $H^n(\ul x_{\tilde p})\nrightarrow H_p(\ul x_{\tilde p})$. Similarly, if \eqref{e:2} fails at some $\tilde p\in [p,1]$, then $H^n(\ol x_{\tilde p})\nrightarrow H_p(\ol x_{\tilde p})$. Finally, if \eqref{e:1} and \eqref{e:2} hold, but \eqref{e:3} fails, then there exists $\varepsilon>0$ such that
\begin{align*}
H^n(\ol x_p)-H^n(\ul x_p)&=\tau (\ul x_p<\chi^n(G)\leq \ol x_p)\geq \tau (G(\ul x_p)=0,\, G(\ol x_p)=1)\\
&\geq \varepsilon>0=H^p(\ol x_p)-H^p(\ul x_p),	
\end{align*}
so $H^n(\ul x_{p})\nrightarrow H_p(\ul x_{p})$ or $H^n(\ol x_{p})\nrightarrow H_p(\ol x_{p})$.
\end{proof}

\bibliographystyle{econometrica}
\bibliography{persuasionlit}

\end{document}